\documentclass[journal,onecolumn]{IEEEtran}
\usepackage{multibib}
\usepackage{ulem}
\usepackage{xcolor}
\usepackage{CJK}
\usepackage{float}
\usepackage[export]{adjustbox}
\usepackage{graphicx}
\graphicspath{ {./images/} }
\usepackage{subcaption}
\usepackage[utf8]{inputenc}
\usepackage[export]{adjustbox}
\usepackage{wrapfig}
\usepackage{authblk}
\usepackage{pgfplots}
\usepackage[top=27mm, bottom=27mm, left=22mm, right=22mm]{geometry}
\usepackage{amsthm,amsmath,amssymb}
\usepackage{mathrsfs}
\usepackage{multirow}
\usepackage{amsthm}
\usepackage{microtype}
\usepackage{hyperref}
\usepackage[capitalize]{cleveref}
\usepackage{xcolor}
\usepackage{verbatim}
\usepackage{indentfirst}
\usepackage{cite}

\usepackage{ulem}

\newcommand{\vu}{\boldsymbol{u}}
\newcommand{\vv}{\boldsymbol{v}}
\newcommand{\vx}{\boldsymbol{x}}
\newcommand{\vy}{\boldsymbol{y}}

\newcommand{\Z}{[q]^n}
\newcommand{\E}{\mathbb{E}}

\newcommand{\ee}{\epsilon}

\newcommand{\beq}{\begin{equation}}
\newcommand{\eeq}{\end{equation}}

\theoremstyle{plain}
\newtheorem{thm}{Theorem}[section]
\newtheorem{lem}[thm]{Lemma}

\newtheorem{prop}[thm]{Proposition}

\newtheorem{rem}[thm]{Remark}

\newtheorem{qu}[thm]{Question}

\newtheorem{claim}[thm]{Claim}

\newcommand{\s}{\subseteq}

\newcommand{\bi}{\binom}
\newcommand{\fr}{\frac}

\title{Improved Gilbert-Varshamov bounds for hopping cyclic codes and optical orthogonal codes\footnote{This project is supported by the National Key Research and Development Program of China under Grant Nos. 2020YFA0712100 and 2018YFA0704703, the National Natural Science Foundation of China under Grant Nos. 11971325, 12231014, and 12101364, the Natural Science Foundation of Shandong Province under Grant No. ZR2021QA005, and the Beijing Scholars Program.}}

\author{Chenyang Zhang\thanks{C. Zhang is with the Research Center for Mathematics and Interdisciplinary Sciences, Shandong University, Qingdao 266237, China (e-mail: chener@mail.sdu.edu.cn).}, Chong Shangguan\thanks{C. Shangguan is with the Research Center for Mathematics and Interdisciplinary Sciences, Shandong University, Qingdao 266237, China, and also with the Frontiers Science Center for Nonlinear Expectations, Ministry of Education, Qingdao 266237, China (e-mail: theoreming@163.com).}, and Gennian Ge\thanks{G. Ge is with the School of Mathematical Sciences, Capital Normal University, Beijing 100048, China (e-mail: gnge@zju.edu.cn).}
}

\begin{document}

\maketitle

\begin{abstract}
  Hopping cyclic codes (HCCs) are (non-linear) cyclic codes with the additional property that the $n$ cyclic shifts of every given codeword are all distinct, where $n$ is the code length. Constant weight binary hopping cyclic codes are also known as optical orthogonal codes (OOCs). HCCs and OOCs have various practical applications and have been studied extensively over the years.

  The main concern of this paper is to present improved Gilbert-Varshamov type lower bounds for these codes, when the          minimum distance is bounded below by a linear factor of the code length. For HCCs, we improve the previously best known lower bound of Niu, Xing, and Yuan by a linear factor of the code length. For OOCs, we improve the previously best known lower bound of Chung, Salehi, and Wei, and Yang and Fuja by a quadratic factor of the code length.  As by-products, we also provide improved lower bounds for frequency hopping sequences sets and error-correcting weakly mutually uncorrelated codes. Our proofs are based on tools from probability theory and graph theory, in particular the McDiarmid's inequality on the concentration of Lipschitz functions and the independence number of locally sparse graphs.
\end{abstract}





\section{Introduction}

Given integers $q$, $n$ and $d$, estimating the maximum size of $q$-ary codes of length $n$ and minimum distance $d$ is a fundamental problem in coding theory. The Gilbert–Varshamov bound (GV bound for short) is a classic lower bound on the size of codes. For fixed $q$, $n\rightarrow\infty$, and $d$ bounded below by a linear factor of $n$, improving upon the GV bound substantially is a well-known difficult task. In this paper we will present improved GV-type bounds for several classes of non-linear cyclic codes.


To move forward let us begin with some needed definitions. For a positive integer $q$, let $[q] =\{0,1,\ldots,q-1\}$. A vector $\vx\in\Z$ is denoted by $\vx=(x_1,\ldots,x_n)$. For two vectors $\vx,\vy\in\Z$, the {\it Hamming distance} $d(\vx,\vy)$ is the number of coordinates where they differ, namely, $d(\vx,\vy)=|\{1\le i\le n:x_i\neq y_i\}|$. A {\it code} of length $n$ and alphabet size $q$ is a subset of $\Z$, whose elements are called {\it codewords}. The {\it minimum distance} $d(C)$ of a code $C\subseteq\Z$ is defined to be $d(C):=\min\{d(\vx,\vy):\vx,\vy\in C,\vx\neq \vy\}$. A $q$-ary code with length $n$, size $M$, and minimum distance $d$ is denoted as an {\it $(n,M,d)_q$-code}.


Next, let us briefly recall some known GV-type bounds. 

\subsection{GV-type bounds}

\paragraph{Generic codes.}  For an integer $1\le t\le n$ and a vector $\vx\in\Z$, let $B(\vx,t):=\{\vy\in\Z:d(\vx,\vy)\le t\}$ denote the {\it Hamming ball} of radius $t$ centered at $\vx$. Given a radius $t$, it can be easily seen that for every $\vx\in\Z$, the volume of the Hamming ball $B(\vx,t)$ is independent of $\vx$, which satisfies that
$$|B(\vx,t)|=\sum_{i=0}^t\bi{n}{i}(q-1)^i=:Vol_q(n,t).$$

The GV bound, proved independently by Gilbert \cite{GV-G} and Varshamov \cite{GV-V}, states that there exist $(n,M,d)_q$-codes with
\begin{align}\label{eqa:GV}
  M\ge\fr{q^n}{Vol_q(n,d-1)}.
\end{align}
For binary codes and $d/n\le 0.499 $, Jiang and Vardy \cite{JV04} improved the lower bound in \eqref{eqa:GV} by a linear factor of $n$. Based on their work, Vu and Wu \cite{VW05} showed that for every $q\ge 2$ and $\tau \le d/n\le 1-1/q-\ee$, where $\tau, \ee$ are some absolute reals, there exist $(n,M,d)_q$-codes with
\begin{align}\label{eqa:VW}
  M=\Omega\left(\fr{nq^n}{Vol_q(n,d-1)}\right).
\end{align}

\paragraph{Non-linear cyclic codes.}
For every $\vx=(x_1,\ldots,x_n)\in \Z$ and $0\le i\le n-1$, we call the vector $\pi_i(\vx):=(x_{i+1},x_{i+2},\ldots,x_{i+n})$\footnote{Throughout this paper, the addition in the subscripts of $x$ is calculated modulo $n$.} the $i$-th {\it cyclic shift} of $\vx$. Let $C(\vx):=\{\pi_i(\vx): 0\le i\le n-1\}$ denote the set consisting of all cyclic shifts of $\vx$. Note that $C(\vx)$ is viewed as a multi-set. A code $C\subseteq\Z$ is said to be a {\it cyclic code} if for every $\vx\in C$, $C(\vx)\s C$.
In the literature, there are a number of works devoting to the constructions of linear cyclic codes with minimum distance ranging from a constant to a sublinear function of $n$, say, $O(\frac{n}{\ln n})$ (see, e.g. \cite[Chapter 8]{roth_2006}). However, when minimum distance is bounded below by a linear factor of $n$, understanding whether there exist {\it asymptotically good} linear cyclic codes is a longstanding open question in coding theory (see \cite{IMME17,DSJP15}).

For {\it non-linear} cyclic codes, the question above has been solved only recently. Haviv, Langberg, Schwartz, and Yaakobi \cite{IMME17} proved that for every prime code length, there exist binary cyclic codes asymptotically attaining the GV bound. Later, Niu, Xing, and Yuan \cite{niu-xing-yuan-fhs} showed that for every $q$ and $n$
there exist  $(n,M,d)_q$-cyclic codes with
\beq\label{eqa: HCC-X}
M\ge \fr{q^n(1-n^2e^{-\fr{\ee^2(\sqrt{n}-2)}{2}})}{Vol_q(n,d-1)-1}.
\eeq

\paragraph{Constant weight codes.}
For a vector $\vx\in\{0,1\}^n$, the {\it weight} $wt(\vx)$ is defined to be the number of its non-zero coordinates, namely, $wt(\vx)=|\{1\le i\le n:x_i\neq 0\}|$. A code is said to be a {\it constant weight code}, if all codewords of it have the same weight. An $(n,M,d)_2$-code with constant weight $w$ is denoted as an $(n,M,d;w)$-code, where we omitted the $q=2$ in the subscript.
For $\vx\in\{0,1\}^n,wt(\vx)=w$, and $0\le t\le 2w$, let $B(\vx,t;w):=\{\vy\in\{0,1\}^n:d(\vx,\vy)\le t, wt(\vy)=w\}$ denote the {\it constant weight Hamming ball} of radius $t$ centered at $\vx$.
Given $t$ and $w$, it is known that for every $\vx\in\{0,1\}^n$ with $wt(\vx)=w$, the volume of $B(\vx,t;w)$ is independent of $\vx$, which satisfies that
$$|B(\vx,t;w)|=\sum_{i=0}^{\left\lfloor t/2\right\rfloor}\bi{w}{i}\bi{n-w}{i}=:Vol(n,t;w).$$

Levenshtein \cite{Levenshtein-CWC} proved a GV-type bound for constant weight codes, showing that there exist $(n,M,d;w)$-codes with
\begin{align}\label{eqa:Le}
  M\ge\frac{\binom{n}{w}}{Vol(n,d-1;w)}.
\end{align}
Similarly to the improvement of  \eqref{eqa:VW} upon \eqref{eqa:GV}, Kim, Liu, and Tran \cite{Kim-Liu-Tran} improved \eqref{eqa:Le} by a linear factor of $n$, showing that for $\tau \le d/n\le (1-\ee)p(1-p)$, where $\tau, \ee, p$ are some absolute reals, there exist $(n,M,d;pn)$-codes with
\begin{align}\label{eqa:KHT}
  M=\Omega\left(\frac{n\binom{n}{pn}}{Vol(n,d-1;pn)}\right).
\end{align}

 \paragraph{Constant weight non-linear cyclic codes.} In the literature, there is also a GV-type lower bound for constant weight non-linear cyclic codes. More precisely, Chung, Salehi, and Wei \cite{Chung89,Chung89-92} and Yang and Fuja \cite{YF95} implicitly showed that there exist $(n,M,d;w)$-cyclic codes with
\beq\label{eqa: GV-cw-cyc}
M\ge \fr{\bi{n}{w}-f(n,w,d)}{n\cdot Vol(n,d-1;w)},
\eeq
where $f(n,w,d)$ is some function of $n,w,d$ (see \cite[Theorem 2]{Chung89} and \cite[Theorem 3]{YF95} for details).

Given the discussion above, it is natural to ask whether one can improve \eqref{eqa: HCC-X} and \eqref{eqa: GV-cw-cyc} by a linear or polynomial factor of $n$, similarly to the improvements of \eqref{eqa:VW} and \eqref{eqa:KHT} made upon \eqref{eqa:GV} and \eqref{eqa:Le}. This is one of the motivating questions of this paper. Note that the method which proves \eqref{eqa:VW} and \eqref{eqa:KHT} cannot be applied directly to improve \eqref{eqa: HCC-X} and \eqref{eqa: GV-cw-cyc}, as one has to take into account the property of cyclic codes. Based on some probabilistic and graph theoretic tools and some ideas from \cite{JV04,VW05,niu-xing-yuan-fhs,Kim-Liu-Tran}, we will answer the above question in a strong sense in Theorems \ref{thm:hc-code} and \ref{thm:bi-cw-hc-code} below.

The main concern of this paper is to present improved GV-type bounds for two special classes of non-linear cyclic codes, namely, hopping cyclic codes and its constant weight counterpart, optical orthogonal codes. As by-products, we also provide improved lower bounds for frequency hopping sequences sets and error-correcting weakly mutually uncorrelated codes as well.

Next, we will introduce these codes and our results in the detail.

\section{Main results}

\subsection{Hopping cyclic codes}

Hopping cyclic codes are (non-linear) cyclic codes with the additional property that the $n$ cyclic shifts of every given codeword are all distinct. Formally speaking, a code $C \subseteq \Z$ is called a {\it hopping cyclic code} (HCC for short) if for every $\vx\in C$, $C(\vx)\subseteq C$, and moreover $C(\vx)$ consists of $n$ distinct elements. We will call a hopping cyclic code $C\subseteq\Z$ an $(n,M,d)_q$-HCC, if it is itself an $(n,M,d)_q$-code. HCCs were originally designed to construct frequency hopping sequences sets (see \cite{Ding09,niu-xing-yuan-fhs} and \cref{subsec:FHS} below), but have found their own interest as an intriguing class of codes. The reader is referred to \cite{niu-xing-yuan-fhs} for more background on HCCs.

For any absolute constant $\ee>0$ and $d<(n-2\sqrt{n})(1-1/q-\ee)$, Niu, Xing, and Yuan (see Theorem III.5 in \cite{niu-xing-yuan-fhs}) showed that there exist $(n,M,d)_q$-HCCs with $M$ bounded below by \eqref{eqa: HCC-X}.

Our first main result improves the above GV-type bound for HCCs (and hence for non-linear cyclic codes) by asymptotically a linear factor of $n$.

\begin{thm}\label{thm:hc-code}
For positive integers $n,q,d$ and absolute reals $\tau,\epsilon\in (0,1-1/q)$ satisfying $\tau \le d/n\le 1-1/q-\ee$, there exist $(n,M,d)_q$-HCCs with
\begin{align*}
    M\ge \fr{cnq^n}{Vol_q(n,d-1)},
\end{align*}
where $c$ is a positive real depending only on $\tau,\ee,q$.
\end{thm}

\subsection{Optical orthogonal codes}

Constant weight binary hopping cyclic codes are also known as optical orthogonal codes (OOCs for short). OOCs enable a large number of asynchronous users to transmit information efficiently and reliably. They have been widely used in various practical scenarios like code-division multiple-access systems and spread spectrum communication. A large number of existing papers were devoted to the constructions of OOCs, see for example \cite{Chung89,Chung13,HI17,LC16,PR13,YF95}.

We will denote an $(n,M,d)_2$-HCC with constant weight $w$ by an $(n,M,d;w)$-OOC, where we omitted the $q=2$ in the subscript. Chung, Salehi, and Wei \cite{Chung89,Chung89-92} and Yang and Fuja \cite{YF95} showed that there exist $(n,M,d;w)$-OOCs with $M$ bounded below by \eqref{eqa: GV-cw-cyc}.

Our second main result improves the above GV-type bound for OOCs (and hence for constant weight non-linear cyclic codes) by asymptotically a quadratic factor of $n$.

\begin{thm}\label{thm:bi-cw-hc-code}
For positive integers $n,q,d$ and absolute reals $p\in (0,1),~\ee\in(0,1/10),~\tau\in(0,p(1-p))$ satisfying $pn\in\mathbb{Z}, ~\tau \le d/n\le (1-\ee)p(1-p)$, there exist $(n,M,d;pn)$-OOCs with
\beq
M\ge \fr{cn\bi{n}{pn}}{Vol(n,d-1;pn)},\notag
\eeq
where $c$ is a positive real depending only on $p,\ee,\tau$.
\end{thm}

\subsection{Sketch of the proofs and comparison with related works}\label{subsec:sketch}

Since the work of Jiang and Vardy \cite{JV04}, it is now well-known that one can lower-bound the size of a code via the lower bound of the independence number of a locally sparse graph defined approximately. Loosely speaking, let $G$ be a graph with vertex set $[q]^n$, where two vertices (vectors) $\vx,\vy\in [q]^n$ are connected by an edge if and only if $d(\vx,\vy)\le d-1$. Then, an independent set $I$ of $G$ corresponds to an $(n,|I|,d)_q$-code. In fact, there is a one-to-one correspondence between independent sets of $G$ and codes in $\Z$ with minimum distance at least $d$. Therefore, given the minimum distance $d$, finding a large code in $\Z$ is equivalent to finding a large independent set in $G$.

Since for every $\vx\in\Z$, there are exactly $Vol_q(n,d-1)-1$ vectors $\vy\in\Z\setminus\{\vx\}$ such that $d(\vx,\vy)\le d-1$, every vertex of $G$ is connected to exactly $Vol_q(n,d-1)-1$ vertices in $\Z$. So, one can construct an independent set in $G$ of size $\frac{q^n}{Vol_q(n,d-1)}$, or equivalently, an $(n,\frac{q^n}{Vol_q(n,d-1)},d)_q$-code, by greedily picking vertices and throwing their neighbors. This gives the GV bound \eqref{eqa:GV}.

Jiang and Vardy \cite{JV04} showed that the graph $G$ defined above is in fact ``locally sparse'' (this will be clarified later in \cref{subsec:indep}). Therefore, there is a more economic way than the greedy algorithm to find a large independent set in $G$ (see \cref{lem: graph} below). This leads to their improvement on the GV bound of generic codes \eqref{eqa:GV} by a linear factor of $n$. Kim, Liu, and Tran \cite{Kim-Liu-Tran} used a similar idea in their proof of \eqref{eqa:KHT}, which improves the GV bound of constant weight codes \eqref{eqa:Le} by a linear factor of $n$.

The high level idea in the proofs of Theorems \ref{thm:hc-code} and \ref{thm:bi-cw-hc-code} will follow the above framework. However, as the codes considered in this paper need not only to have large minimum distance but also to be hopping cyclic, the corresponding graphs are not as neat as the graph $G$ described above.

The proofs of Theorems \ref{thm:hc-code} and \ref{thm:bi-cw-hc-code} can be summarized as follows:


\begin{enumerate}
    \item [Step 1.] define an appropriate graph $G_{HCC}$ (resp. $G_{OOC}$) whose independent sets correspond to HCCs (resp. OOCs);   
    \item [Step 2.] show that $G_{HCC}$ (resp. $G_{OOC}$) has sufficiently many vertices; 
    \item [Step 3.] show that $G_{HCC}$ (resp. $G_{OOC}$) is in fact locally sparse; 
    \item [Step 4.] use known lower bound on the independence number of locally sparse graphs to show that $G_{HCC}$ (resp. $G_{OOC}$) has sufficiently large independent sets, and hence there exist sufficiently large HCCs (resp. OOCs).
\end{enumerate}

Next, we would like to compare our work with some previous papers \cite{JV04,VW05,niu-xing-yuan-fhs,Kim-Liu-Tran}. On one hand, in \cite{JV04,VW05,Kim-Liu-Tran}, $G$ and its constant weight counterpart were used to prove improved GV-type bounds for generic codes and constant weight codes, respectively. $G_{HCC}$ (resp. $G_{OOC}$) defined in this paper is quite different from $G$ (resp. its constant weight counterpart). To be more precise, for $\vx\in\Z$ let
 \begin{align}\label{eqa:d(x)}
   d(\vx)=\min\{d_H(\pi_i(\vx),\pi_j(\vx)):0\le i<j\le n-1\}=\min\{d_H(\vx,\pi_i(\vx)):1\le i\le n-1\}
 \end{align}
 denote the minimum distance $d(C(\vx))$. By definition, $C\s\Z$ is an HCC only if for every $\vx\in C$, $d(\vx)\ge 1$ and moreover, $C$ is an $(n,M,d)_q$-HCC only if for every $\vx\in C$, $d(\vx)\ge d$. Therefore, compared with $G$ whose vertex set is $\Z$, in the vertex set of $G_{HCC}$ we have to rule out all vectors $\vx$ with $d(\vx)<d$. In fact, the vertex set of $G_{HCC}$ is not vectors in $\Z$ but some well-defined subsets of vectors in $\Z$. Hence, to execute Steps 1 and 2, we have to show that for many vectors $\vx\in\Z$, $d(\vx)$ is quite large. Similar discussion works for $G_{OOC}$. The details can be found in Sections \ref{sec qhc} and \ref{sec cw-hc} below.

 On the other hand, to show that $G_{HCC}$ and $G_{OOC}$ are locally sparse, we will use some well-developed tools that were continually improved in \cite{JV04,VW05,Kim-Liu-Tran}. In particular, we will use two upper bounds on the intersection volume of Hamming balls, which were proved in \cite{Kim-Liu-Tran} to show that $G$ and its constant weight counterpart are locally sparse.

 Lastly, let us compare our work with \cite{niu-xing-yuan-fhs}. Although \cite{niu-xing-yuan-fhs} is not graph theoretic, an important step in their proof of \eqref{eqa: HCC-X} was also to show that for $q^n-o(q^n)$ vectors $\vx\in\Z$, $d(\vx)$ is quite large (see \cite[Lemma III.4]{niu-xing-yuan-fhs}). To do so, they used some standard concentration inequality for martingales. In this paper, we used the McDiarmid's inequality instead, which makes our proof easier and gives a better control (upper bound) on the lower order term $o(q^n)$ (see \cref{rem:NXY} below).

A summary of previous mentioned GV-type bounds is presented in \cref{table}.

\begin{table*}[h]
 \renewcommand{\arraystretch}{2.5}
  \caption{A summary of GV-type bounds}\label{table}
  \centering
  \begin{tabular}{|c|c|c|}
    \hline
     & GV-type lower bounds of $M$ & Improved GV-type lower bounds of $M$ \\\hline
    Generic $(n,M,d)_q$-codes & $\fr{q^n}{Vol_q(n,d-1)}$ \eqref{eqa:GV} ( see \cite{GV-G,GV-V})  & $\Omega\left(\fr{nq^n}{Vol_q(n,d-1)}\right)$ \eqref{eqa:VW} (see \cite{JV04,VW05}) \\\hline
   Non-linear $(n,M,d)_q$-cyclic codes & $\fr{(1-o(1))q^n}{Vol_q(n,d-1)}$ \eqref{eqa: HCC-X} (see \cite{HI17,niu-xing-yuan-fhs}) & $\fr{cnq^n}{Vol_q(n,d-1)}$ \cref{thm:hc-code} (this paper) \\\hline
    Constant weight $(n,M,d;pn)$-codes & $\frac{\binom{n}{pn}}{Vol(n,d-1;pn)}$ \eqref{eqa:Le} ( see \cite{Levenshtein-CWC}) & $\Omega\left(\frac{n\binom{n}{pn}}{Vol(n,d-1;pn)}\right)$ \eqref{eqa:KHT} (see \cite{Kim-Liu-Tran}) \\\hline
    Constant weight non-linear $(n,M,d;pn)$-cyclic codes &
$\fr{\bi{n}{w}-f(n,d,w)}{n\cdot Vol(n,d-1;w)}$ \eqref{eqa: GV-cw-cyc} (see \cite{Chung89,Chung89-92,YF95})
 & $\fr{cn\bi{n}{pn}}{Vol(n,d-1;pn)}$ \cref{thm:bi-cw-hc-code} (this paper) \\
    \hline
  \end{tabular}
\end{table*}

The rest of this paper is organized as follows. In \cref{sec:appl} we will mention two applications of \cref{thm:hc-code}, which gives new lower bounds for frequency hopping sequences sets and error-correcting weakly mutually uncorrelated codes. In \cref{sec:pre} we will collect the tools that are used in Steps 1-4. In \cref{subsec:indep} we will formally define locally sparse graphs (this will be used in Step 1) and state a lower bound on its independence number (this will be used in Step 4, see \cref{lem: graph} below). In \cref{subsec:mc} we will introduce the McDiarmid's inequality (this will be used in Step 2, see \cref{thm: mc} below). In \cref{subsec:vol} we will state two upper bounds on the intersection volume of Hamming balls (this will be used in Step 3, see Lemmas \ref{lem:vol q} and \ref{lem:vol cw} below).
The proofs of Theorems \ref{thm:hc-code} and \ref{thm:bi-cw-hc-code} will be presented in Sections \ref{sec qhc} and \ref{sec cw-hc} respectively. We will conclude this paper in \cref{sec:con}.



\section{Applications of the main results}\label{sec:appl}

\subsection{Frequency hopping sequences}\label{subsec:FHS}

In this subsection, we will present an application of Theorem \ref{thm:hc-code} to frequency hopping sequences. We will need the following easy lemma.


\begin{lem}\label{lem:equivalent-class}
\begin{itemize}
   \item [{\rm (i)}] For every $\vx,\vy\in\Z$, $\vx\in C(\vy)$ if and only if $\vy\in C(\vx)$; moreover, $C(\vx)=C(\vy)$ if and only if $\vx\in C(\vy)$.
   \item [{\rm (ii)}] Let $C$ be an $(n,M,d)_q$-HCC. Then $n\mid M$, and $C$ can be partitioned into a pairwise disjoint union of $M/n$ $C(\vx)$'s, where $\vx\in C$.
\end{itemize}
\end{lem}

\begin{proof}
Note that (i) follows straightforwardly from the definition of $C(\vx)$. To prove (ii), it is not hard to check that the relation $\sim$ on the set $C$ defined by $\vx\sim\vy$ if $C(\vx)=C(\vy)$ is an equivalence relation. By (i) and the definition of an HCC, each equivalent class equals to $C(\vx)$ for some $\vx\in C$, and hence consists of $n$ distinct vectors. Therefore, (ii) follows from the fact that the equivalent classes given by $\sim$ form a partition of $C$.
\end{proof}

Let $C$ be an $(n,M,d)_q$-HCC. Given the lemma above, consider the partition of $C$ formed by the equivalent classes given by $\sim$. If we choose a representative element from each of these equivalent classes, then these representatives form a set of {\it frequency hopping sequences} (FHSs for short).
An FHS set obtained by an $(n,M,d)_q$-HCC in the above manner is denoted as an $(n,M/n,n-d)_q$-FHS, which is a set $F\s\Z$ with $|F|=M/n$ such that
\begin{itemize}
  \item for every $\vx\in F$ and $0< i\le n-1$, $H_{\vx,\vx}(i):=n-d_H(\vx,\pi_i(\vx))\le n-d$;
  \item for every distinct $\vx,\vy\in F$ and $0\le i\le n-1$, $H_{\vx,\vy}(i):=n-d_H(\vx,\pi_i(\vy))\le n-d$;
\end{itemize}
where $H_{\vx,\vx}(i)$ and $H_{\vx,\vy}(i)$ are known as the {\it auto- and cross-Hamming correlation functions} at time delay $i$ (see \cite{Ding09,Ding10,niu-xing-yuan-fhs}). FHSs were designed for the transmission of radio signals and have been studied extensively. 
In fact, the work of Niu et al. \cite{niu-xing-yuan-fhs} was motivated by the study of FHSs.

Note that one can also define FHSs directly using the Hamming correlation functions. For the sake of saving space, we will not formally state such a definition. It was observed in \cite{Ding09,niu-xing-yuan-fhs}
that the two definitions of FHSs are equivalent, since there is a one-to-one correspondence between HCCs and FHSs. Based on such a correspondence and \eqref{eqa: HCC-X}, Niu et al. (see \cite[Corollary III.6]{niu-xing-yuan-fhs}) showed that for any absolute constant $\ee>0$ and $\lambda\ge n-(n-2\sqrt{n})(1-1/q-\ee)$ there exist $(n,M,\lambda)_q$-FHSs with
\beq\label{eqa: FHS}
M\ge \fr{q^n(1-n^2e^{-\fr{\ee^2(\sqrt{n}-2)}{2}})}{n(Vol_q(n,n-\lambda-1)-1)}.
\eeq

The following result improves \eqref{eqa: FHS} by a linear factor of $n$.

\begin{prop}\label{prop:fhs}
For positive integers $n,q,\lambda$ and absolute reals $\tau,\epsilon\in (0,1-1/q)$ satisfying $1/q+\ee \le \lambda/n\le 1-\tau$, there exist $(n,M,\lambda)_q$-FHSs with
\begin{align*}
    M\ge \fr{cq^n}{Vol_q(n,n-\lambda-1)},
\end{align*}
where $c$ is a positive real depending only on $\tau,\ee,q$.
\end{prop}

\begin{proof}
  This is an easy consequence of \cref{lem:equivalent-class} (ii), \cref{thm:hc-code},  and the definition of FHSs.
\end{proof}

\subsection{Error-correcting weakly mutually uncorrelated codes}

For $1\le\kappa\le n$, a code $C\in\Z$ is said to be a $\kappa$-{\it weakly mutually uncorrelated code} (WMUC for short) if for all $\kappa\le\ell\le n-1$, no proper {\it prefix} of length $\ell$ of a codeword in $C$ appears as a {\it suffix} of another codeword in $C$, including itself. Formally speaking, $C$ is a $\kappa$-WMUC if for every $\vx,\vy\in C$ (possibly identical) and every $\kappa\le\ell\le n-1$, $(x_1,\ldots,x_\ell)\neq(y_{n-\ell+1},\ldots,y_n)$. A WMUC is said to be {\it error-correcting} if it has some sufficiently large minimum distance. 

WMUCs and 
error-correcting WMUCs were introduced by Yazdi, Kiah, Gabrys, and Milenkovic \cite{Yazdi-WMU} as a technique that allows random access to encoded DNA strands in DNA-based storage systems. The reader is referred to \cite{Yazdi-WMU} for more background. 

The next result shows that HCCs have the bonus of being also WMUC.

\begin{prop}
For positive integers $n,q,\kappa$ and absolute reals $\tau,\epsilon\in (0,1-1/q)$ satisfying $1/q+\ee \le \kappa/n\le 1-\tau$ there exist $\kappa$-WMU $(n,M,n-\kappa+1)_q$-codes with
  \begin{align*}
    M\ge \fr{cq^n}{Vol_q(n,n-\kappa)},
\end{align*}
where $c$ is a positive real depending only on $\tau,\ee,q$.
\end{prop}

\begin{proof}
  By \cref{thm:hc-code}, under the assumption of the proposition, there exists an $(n,M',n-\kappa+1)_q$-HCC $C'\s\Z$ with $M'\ge \fr{cnq^n}{Vol_q(n,n-\kappa)},$ where $c$ is a positive real depending only on $\tau,\ee,q$. By \cref{lem:equivalent-class} (ii), $C'$ can be partitioned into a pairwise disjoint union of $M'/n$ $C(\vx)$'s. Just as the proof of \cref{prop:fhs}, let $C$ be a subcode of $C'$ formed by picking exactly one codeword from every $C(\vx)$ in the partition. Then, $C$ is an $(n,M,n-\kappa+1)_q$-code with $M\ge \fr{cq^n}{Vol_q(n,n-\kappa)}$.

  It remains to show that $C$ is also $\kappa$-WMU. Assume to the contrary that there exist $\vx,\vy\in C$ (possibly identical) such that for some $\kappa\le\ell\le n-1$, the prefix $(x_1,\ldots,x_{\ell})$ is identical to the suffix $(y_{n-\ell+1},\ldots,y_n)$. Note that $\{\vx,\vy\}\s C\s C'$. As $C$ is an HCC, $\pi_{n-\ell}(y)=(y_{n-\ell+1},\ldots,y_n,y_1,\ldots,y_{n-\ell})\in C'$. As $d_H(x,\pi_{n-\ell}(y))\le n-\ell\le n-\kappa$, which is strictly less than the minimum distance of $C'$, we must have $x=\pi_{n-\ell}(y)$. It implies that $\{\vx,\vy\}\s C(\vx)$ and $|C\cap C(\vx)|\ge 2$, a contradiction.
\end{proof}

\section{Collection of some useful lemmas}\label{sec:pre}

In this section, we will collect the tools that are used in this paper.



\subsection{Locally sparse graphs and their independent sets}\label{subsec:indep}

In this subsection, we will introduce locally sparse graphs and a lower bound on their independence number. A graph $G=(V,E)$ consists of a vertex set $V$ and an edge set $E$, where $V$ is a finite set and $E$ is a family of $2$-subsets of $V$. Two vertices $u,v\in V$ are said to be connected by an edge if $\{u,v\}\in E$. Two connected vertices are called {\it neighbors}. The set of neighbors of a vertex $u$ in $G$ is denoted by $N_G(u)$. The {\it degree} of a vertex $u\in V$ in $G$, denoted by $\deg_G(u)$, is the size of its neighborhood, i.e., $\deg_G(u)=|N_G(u)|$. Note that we will omit the subscript $G$ in $N_G(u)$ and $\deg_G(u)$ when the underlying graph $G$ is clear from the context.
The {\it maximum degree} of a graph is the maximum degree among all of its vertices.
For a subset $S\subseteq V$, the set of edges {\it induced} by $S$ consists of all edges with both endpoints in $S$, namely the edges $\{\{u,v\}\in E:u,v\in S\}$. A graph is called {\it locally sparse} if for every vertex, its neighborhood induces only a limit number of edges.

Given a graph $G=(V,E)$, a subset $I\subseteq V$ is called an {\it independent set} of $G$ if every two vertices in $I$ are not connected. The {\it independence number} of $G$, denoted by $\alpha(G)$, is the size of the maximum independent set of $G$. Using the greedy algorithm mentioned in \cref{subsec:sketch}, it is not hard to see that every graph with maximum degree $D$ has an independent set of size at least $\frac{|V|}{D+1}$. \cref{lem: graph} below, which is a direct consequence of \cite[Corollary 1]{Pirot-Hurley}, presents a better lower bound for the independence number of locally sparse graphs. 

%




\begin{lem}\cite[Corollary 1]{Pirot-Hurley}\label{lem: graph} 
Let $G$ be a graph on $n$ vertices with maximum degree at most $D$, where $D\to\infty$ as $n\to\infty$. Suppose that for every vertex of $G$, its neighborhood induces at most $\fr{D^2}{K}$ edges, where $1\le K \le D^2+1$. Then
$$\alpha(G)\ge(1-o(1))\cdot\fr{|V|}{D}\cdot\ln(\min\{D,K\}),$$
where $o(1)\to 0$ as $K\to\infty$.
\end{lem}

\subsection{McDiarmid's inequality}\label{subsec:mc}

We will make use of the powerful McDiarmid's inequality, which, roughly speaking, states that a Lipschitz function of random variables is concentrated around its expectation.

\begin{lem}{\cite{mc89}}\label{thm: mc}
(McDiarmid's inequality). Given sets $\mathcal{X}_1,\ldots,\mathcal{X}_n$, let $X_1, \ldots, X_n$ be independent random variables taking values in  $\mathcal{X}_1,\ldots,\mathcal{X}_n$ respectively. Let $f:\mathcal{X}_1 \times  \cdots \times \mathcal{X}_n \to \mathbb{R}$ be a mapping. If there exist constants $c_1,\ldots, c_n\in\mathbb{R}$ such that for each $1\le i\le n$,
\begin{align*}
    \sup \limits_{x_i'\in \mathcal{X}_i, x_1\in \mathcal{X}_1, \ldots, x_n\in \mathcal{X}_n}|f(x_1,\ldots,x_i,\ldots,x_n)-f(x_1,\ldots,x_{i}',\ldots,x_n)| \le c_i,
\end{align*}
(i.e., changing the value of the $i$-th coordinate $x_i$ change the absolute value of $f$ by at most $c_i$), then for any $t >0$,
\begin{align*}
    \Pr[f(X_1,X_2, \ldots,X_n)-\E[f(X_1,X_2, \ldots, X_n)] \le -t] \le \exp \left(- \frac{2 t^2}{\sum_{i=1}^n c_i^2} \right).
\end{align*}
\end{lem}

\subsection{The intersection volume of two Hamming balls}\label{subsec:vol}

We will need the following estimation on the intersection volume of two Hamming balls.

\begin{lem}{\cite[Lemma 4.2]{Kim-Liu-Tran}}\label{lem:vol q}
For positive integers $n,q,d$ and two absolute reals $\tau,\epsilon\in (0,1-1/q)$ satisfying $\tau \le d/n\le 1-1/q-\ee$, we have that
\begin{align*}
    \fr{Vol_q(n,d-\tau n/2)}{Vol_q(n,d-1)}\le e^{-\Omega_{\tau,q}(1)\cdot n}.
\end{align*}
Moreover, for every $\vx,\ \vy \in \Z$,
\begin{align*}
    \fr{|B(\vx,d-1)\cap B(\vy,d-1)|}{Vol_q(n,d-1)}=e^{-\Omega_{\tau,q}(1)\cdot d(\vx,\vy)}.
\end{align*}
\end{lem}

Similar upper bounds can also be proved for the intersection volume of two constant weight Hamming balls. 

\begin{lem}{\cite[Lemma 4.3]{Kim-Liu-Tran}}\label{lem:vol cw}
Given positive integers $n,d$, and reals $p\in(0,1),~\ee\in(0,1/10),~\tau\in (0,p(1-p))$ satisfying $\tau \le d/n\le (1-\ee)p(1-p)$. Then
\begin{align*}
     \fr{Vol(n,d-\tau n/2;pn)}{Vol(n,d-1;pn)}\le 2e^{-\Omega_{\ee,\tau}(1)n}.
\end{align*}
Moreover, for every $\vx,\ \vy \in \{0,1\}^n$ with Hamming weight $pn$,
\begin{align*}
    \fr{|B(\vx,d-1;pn)\cap B(\vy,d-1;pn)|}{Vol(n,d-1;pn)}=2e^{-\Omega_{\ee}(1)\cdot (d(\vx,\vy)+d(\vx,\vy)^2 / (d-1))}.
\end{align*}
\end{lem}

\section{Proof of \cref{thm:hc-code}}\label{sec qhc}

Given two reals $\tau,\ee\in(0,1-1/q)$, the goal of \cref{thm:hc-code} is to show the existence of large $(n,M,d)_q$-HCCs whenever $\tau \le d/n\le 1-1/q-\ee$. The proof of \cref{thm:hc-code} will follow the steps outlined in \cref{subsec:sketch}. First of all, note that an $(n,M,d)_q$-code $C$ is an $(n,M,d)_q$-HCC only if for every $\vx\in C$, $d(\vx)\ge d$ (recall \eqref{eqa:d(x)}). Therefore, to construct large HCCs, we need to show the existence of a large set $A\subseteq\Z$ such that for each $\vx\in A$, $d(\vx)\ge d$. Such a set $A$ in fact exists, as shown by the next lemma.

\begin{lem}\label{lem: dle}
Let $n,q$ be positive integers and  $\epsilon\in (0,1-1/q)$ be a real. Let $$A=\{\vx\in\Z:d(\vx)>n(1-1/q-\ee)\}.$$ Then $$|A|\ge q^n\left(1-(n-1)\exp(-\frac{\ee^2 n}{2})\right).$$
\end{lem}

\begin{proof} Let $X=(X_1,\ldots,X_n)$ be a uniformly chosen random element of $\Z$.  To prove the lemma, it is enough to show that
\begin{align}\label{eqa:d(X)}
\Pr[d(X) \le n(1-1/q-\ee)] \le (n-1)\exp(-\frac{\ee^2 n}{2}).
\end{align}
Recall that $d(X)=\min\{d(X,\pi_i(X)):1\le i\le n-1\}$, where $\pi_i(X)=(X_{i+1},\ldots,X_{i+n})$. By the union bound,
\begin{align*}
    \Pr[d(X) \le n(1-1/q-\ee)]&=\Pr[\exists~1\le i\le n-1~\text{s.t.}~d(X,\pi_i(X))\le n(1-1/q-\ee)]\\
    &\le\sum_{i=1}^{n-1} \Pr[d(X,\pi_i(X))\le n(1-1/q-\ee)].
\end{align*}

According to the discussion above, to prove \eqref{eqa:d(X)} it suffices to show that for every fixed $1\le i\le n-1$,
$$\Pr[d(X, \pi_i(X)) \le n(1-1/q-\ee)]\le\exp(-\frac{\ee^2 n}{2}).$$ Note that choosing $X=(X_1,\ldots,X_n) \in \Z$ uniformly at random is equivalent to choosing $X_1,\ldots,X_n\in [q]$ uniformly and independently at random. Let $$\delta:[q]\times[q]\rightarrow\{0,1\}$$ denote the Kronecker function such that for $a,b\in[q]$, $\delta(a,b)=1$ if $a=b$ and $\delta(a,b)=0$ if $a\neq b$. It is straightforward to check by definition that for every $\vx\in[q]^n$,
\begin{align}\label{eqa:d}
    d(\vx,\pi_i(\vx))=n-\sum_{j=1}^n \delta(x_j,x_{j+i}).
\end{align}

We will apply \cref{thm: mc} with $f(x_1,\ldots,x_n):=d(\vx,\pi_i(\vx))$. Observe that for every $1\le j\le n$, changing the value of $x_j$ could only change the values of $\delta(x_{j-i},x_j)$ and $\delta(x_j,x_{j+i})$ in the right hand side of \eqref{eqa:d}, which in turn changes the absolute value of $d(\vx,\pi_i(\vx))$ and hence $f$ by at most $2$. In other words, for every $1\le j\le n$,
\begin{align}\label{eqa:lipchitz}
    \sup_{x'_j,x_1,\ldots,x_n\in[q]} |f(x_1,\ldots,x_j,\ldots,x_n)-f(x_1,\ldots,x'_j,\ldots,x_n)|\le 2.
\end{align}
Moreover, by the linearity of expectation, it is easy to see that for every $1\le i\le n-1$ and $1\le j\le n$, $\E(\delta(X_j,X_{j+i}))=1/q$, which implies that $\E[f(X_1,\ldots,X_n)=n(1-1/q)$. Applying \cref{thm: mc} to $f$, we obtain
\begin{align*}
    \Pr\left[d(X,\pi_i(X))-n(1-1/q) \le -\ee n\right]\le \exp(-\fr{\ee^2n}{2}).
\end{align*}
Plugging the above inequality to the aforementioned union bound proves \eqref{eqa:d(X)}, and hence \cref{lem: dle}.
\end{proof}

\begin{rem}\label{rem:NXY}
  Using a standard concentration inequality for martingales, \cite[Lemma III.4]{niu-xing-yuan-fhs} proved a slightly weaker result
  \begin{align*}
    |\{\vx\in\Z:d(\vx)>(n-2\sqrt{n})(1-1/q-\ee)\}|\ge q^n\left(1-n^2e^{-\fr{\ee^2(\sqrt{n}-2)}{2}}\right).
  \end{align*}
\end{rem}

\begin{rem}\label{rem:large-V}
 It is not hard to check by definition that $A$ is an HCC with minimum distance, say 1. Therefore, similarly to the proof of \cref{lem:equivalent-class} (ii), one can show that with the equivalence relation $\sim$, $A$ can be partitioned into a family $V$ of equivalent classes such that every equivalent class consists of exactly $n$ distinct vectors of $\Z$.
 It follows by \cref{lem: dle} that
\begin{align}\label{eqa:|V|}
    |V|=\frac{|A|}{n}\ge \frac{1}{n}\cdot q^n\left(1-(n-1)\exp(-\fr{\ee^2n}{2})\right).
\end{align}
Moreover, by the definition of $A$ we have that for every $C(\vx)\in V$,
\begin{align}\label{eqa:dx}
    d(\vx)>n(1-1/q-\ee),
\end{align}
and for every distinct $C(\vx),C(\vy)\in V$,
\begin{align}\label{eqa:disjoint}
    C(\vx)\cap C(\vy)=\emptyset.
\end{align}
\end{rem}

To prove \cref{thm:hc-code}, we will construct an auxiliary graph on the vertex set $V$, and then apply \cref{lem: graph} to this graph, as detailed below.

\begin{lem}\label{lem:implyThm1}
  With the notation of \cref{thm:hc-code}, fixing some $d\in[\tau n, n(1-1/q-\ee)]$, let $G_{HCC}=(V,E)$ be a graph constructed as below, where $V$ is the family of equivalent classes that partition $A$, and two distinct vertices (or equivalent classes) $C(\vx),C(\vy)\in V$ form an edge in $E$ if and only if
  $$d(C(\vx),C(\vy))=\min\{d(\vx_i,\vy_j):0\le i,j\le n-1\}\le d-1.$$ Then the following holds:
  \begin{itemize}
      \item [{\rm (i)}] for every independent set $I\subseteq V$, the set of vectors $ C:=\bigcup_{C(\vx)\in I} C(\vx)$ forms an $(n,n|I|,d)_q$-HCC;
      \item [{\rm (ii)}] the maximum degree of $G_{HCC}$ is at most $D:=Vol_q(n,d-1)$;
      \item [{\rm (iii)}] for every vertex of $G_{HCC}$, its neighborhood induces at most $\frac{D^2}{K}$ edges, where $K=e^{\Omega_{\tau,q}(1)\cdot n}$.
  \end{itemize}
\end{lem}

\begin{proof}
    To prove (i), note first that for every $C(\vx)\in V$, $C(\vx)$ consists of $n$ distinct vectors, which implies that $C$ is indeed an HCC. Moreover, it follows by \eqref{eqa:disjoint} that  $|C|=n|I|$. Lastly, to prove that $d(C)\ge d$, let $\boldsymbol{c}^1$, $\boldsymbol{c}^2\in C$ be two distinct codewords. On one hand, if $\{\boldsymbol{c}^1,\boldsymbol{c}^2\}\subseteq C(\vx)$ for some $C(\vx)\in I\subseteq V$, then by \eqref{eqa:dx} we have
\begin{align*}
 d(\boldsymbol{c}^1,\boldsymbol{c}^2)\ge d(\vx)>n(1-1/q-\ee)\ge d.
\end{align*}
On the other hand, if $\boldsymbol{c}^1\in C(\vx)$ and $\boldsymbol{c}^2\in C(\vy)$ for distinct $C(\vx),C(\vy)\in I$, then as $I$ is an independence set in $G$, we have
\begin{align*}
    d(\boldsymbol{c}^1,\boldsymbol{c}^2)\ge d(C(\vx),C(\vy))\ge d,
\end{align*}
completing the proof of (i).

To prove (ii), note that for every $C(\vx)\in V$,
\begin{align*}
    \deg(C(\vx))&=|\{C(\vy)\in V\setminus\{C(\vx)\}: d(C(\vx),C(\vy))\le d-1\}|\\
    &=|\{C(\vy)\in V\setminus\{C(\vx)\}: \min_{0\le i\le n-1}d(\vx,\pi_i(\vy))\le d-1\}|\\
    &=|\{C(\vy)\in V\setminus\{C(\vx)\}: \exists~i\in[n]~\text{s.t.}~\pi_i(\vy)\in B(\vx,d-1)\}|\\&\le |B(\vx,d-1)|=D,
\end{align*}
as needed.

To prove (iii), fixing a vertex $C(\vx) \in V$, let $\Gamma$ denote the subgraph induced by the neighborhood of $C(\vx)$. Partition $V(\Gamma)=S\cup T$, where $$S=\{C(\vy)\in V(\Gamma):d(C(\vx),C(\vy))\le d-\tau n/2\}$$ and $$T=\{C(\vy)\in V(\Gamma):d-\tau n/2< d(C(\vx),C(\vy))\le d-1\}.$$ We have the following claim.
\begin{claim}\label{claim:sparse-graph-1}
  There exists some $K=e^{\Omega_{\tau,q}(1)\cdot n}$ such that $|S|\le D/K$, and for every vertex $C(\vy)\in T$, $\deg_{\Gamma}(C(\vy))\le D/K$.
\end{claim}
Given the correctness of the claim, it follows that
\begin{align*}
  |E(\Gamma)|&=\fr{1}{2}\left(\sum_{C(\vy)\in S}\deg_{\Gamma}(C(\vy))+\sum_{C(\vy)\in T}\deg_{\Gamma}(C(\vy))\right)\\
  &\le\fr{1}{2}\left(|S|\cdot D+|T|\cdot \fr{D}{K}\right)\le\fr{D^2}{K},
\end{align*}
completing the proof of  (iii).
\end{proof}

It remains to prove \cref{claim:sparse-graph-1}.

\begin{proof}[Proof of \cref{claim:sparse-graph-1}]
To prove the first inequality, observe that
\begin{align*}
    |S|&=|\{C(\vy)\in V(\Gamma):d(C(\vx),C(\vy))\le d-\tau n/2\}|\\
    &=|\{C(\vy)\in V(\Gamma):\min_{0\le i\le n-1}d(\vx,\pi_i(\vy))\le d-\tau n/2\}|\\
    &=|\{C(\vy)\in V(\Gamma):\exists~i\in[n]~\text{s.t.}~\pi_i(\vy)\in B(\vx,d-\tau n/2)\}|\\
    &\le |B(\vx,d-\tau n/2)|=Vol_q(n,d-\tau n/2).
\end{align*}
Therefore,
\begin{align}\label{eqa:S}
 \fr{|S|}{D}\le\fr{Vol_q(n,d-\tau n/2)}{Vol_q(n,d-1)}\le e^{-\Omega_{\tau,q}(1)\cdot n},
\end{align}
where the last inequality follows by the first equation of \cref{lem:vol q}.

To prove the second inequality, note that for every $C(\vy)\in T$, we have
\begin{align*}
N_{\Gamma}(C(\vy))
&=N(C(\vx))\cap N(C(\vy))\\
&=\{C(\vv)\in V\setminus\{C(\vx),C(\vy)\}:d(C(\vv),C(\vx))\le d-1,~d(C(\vv),C(\vy))\le d-1\}\\
&=\{C(\vv)\in V\setminus\{C(\vx),C(\vy)\}:\min_{0\le i\le n-1}d(\vv,\pi_i(\vx))\le d-1, \min_{0\le j\le n-1}d(\vv,\pi_j(\vy))\le d-1\}\\
&=\{C(\vv)\in V\setminus\{C(\vx),C(\vy)\}:\exists~0\le i,j\le n-1~\text{s.t.}~\vv\in B(\pi_i(\vx),d-1)\cap B(\pi_j(\vy),d-1)\},\\
&\subseteq\big\{C(\vv)\in V:\vv\in\bigcup_{0\le i,j\le n-1} B(\pi_i(\vx),d-1)\cap B(\pi_j(\vy),d-1)\big\},
\end{align*}
which implies that
\begin{align*}
    |N_{\Gamma}(C(\vy))|\le\sum_{0\le i,j\le n-1} |B(\pi_i(\vx),d-1)\cap B(\pi_j(\vy),d-1)|.
\end{align*}
It thus follows that
\begin{equation}\label{eqa:T}
  \begin{aligned}
  \fr{\deg_{\Gamma}(C(\vy))}{D}&=\fr{|N_{\Gamma}(C(\vy))|}{Vol_q(n,d-1)}\le\fr{\sum_{0\le i,j\le n-1} |B(\pi_i(\vx),d-1)\cap B(\pi_j(\vy),d-1)|}{Vol_q(n,d-1)}\\
  &=\sum_{0\le i,j\le n-1}\frac{B(\pi_i(\vx),d-1)\cap B(\pi_j(\vy),d-1)}{Vol_q(n,d-1)}\\
  &\le\sum_{0\le i,j\le n-1} e^{-\Omega_{\tau,q}(1)\cdot d(\pi_i(\vx),\pi_j(\vy))}\\
  &\le n^2\cdot  e^{-\Omega_{\tau,q}(1)\cdot n}=e^{-\Omega_{\tau,q}(1)\cdot n},
\end{aligned}
\end{equation}
where the last two inequalities follows from the second equation of \cref{lem:vol q} and the assumption that $d(\pi_i(\vx),\pi_j(\vy))\ge d(C(\vx),C(\vy))\ge d\ge \tau n$.

The claim follows straightforwardly by combining \eqref{eqa:S} and \eqref{eqa:T}.
\end{proof}

We proceed to present the proof of \cref{thm:hc-code}.

\begin{proof}[Proof of \cref{thm:hc-code}]
Let $G_{HCC}=(V,E)$ be the graph defined in \cref{lem:implyThm1}. Applying \cref{lem: graph} in concert with \eqref{eqa:|V|} and \cref{lem:implyThm1} (ii), (iii) yields that
\begin{align*}
  \alpha(G_{HCC})&\ge(1-o(1))\cdot\fr{|V|}{D}\cdot\ln(\min\{D,K\})\\
  &\ge (1-o(1))\cdot \frac{q^n}{n} \cdot\fr{1}{Vol_q(n,d-1)}\cdot\ln(\min\{Vol_q(n,d-1),e^{\Omega_{\tau,q}(n)}\})\\
  &\ge \fr{cq^n}{Vol_q(n,d-1)},
\end{align*}
where $c$ is a constant independent of $n$. Therefore, according to \cref{lem:implyThm1} (i), there exist $(n,M,d)_q$-HCCs with $M\ge n\cdot\alpha(G_{HCC})$, completing the proof of the theorem.
\end{proof}

\section{Proof of \cref{thm:bi-cw-hc-code}}\label{sec cw-hc}

Given reals $p\in(0,1),~\ee\in(0,1/10),~\tau\in(0,p(1-p)$, the aim of \cref{thm:bi-cw-hc-code} is to show the existence of large $(n,M,d;pn)$-OOCs whenever $\tau \le d/n\le (1-\ee)p(1-p)$. For $pn\in\mathbb{Z}_+$, let $\bi{[n]}{pn}=\{\vx\in\{0,1\}^n:wt(\vx)=pn\}$. Note that an $(n,M,d;pn)$-code $C$ is an $(n,M,d;pn)$-OOC only if for every $\vx\in C$, $d(\vx)\ge d$. Similarly to the discussion at the beginning of \cref{sec qhc}, to construct large OOCs we need to show the existence of a large set $B\subseteq\bi{[n]}{pn}$ such that for each $\vx\in B$, $d(\vx)\ge d$. The proof of this result is in the spirit similar to that of \cref{lem: dle} but technically a bit more involved.

\begin{lem}\label{lem: dle-cw-b}
Let $\ee>0,~0<p<1$ be reals and $n$ be a positive integer.
Let $$B=\{\vx\in\bi{[n]}{pn}:d(\vx)>(1-\ee)np(1-p)\}.$$ Then $$|B|\ge \bi{n}{pn}\left(1-n^{3/2}\exp\big(-\fr{\Omega_{p,\ee}(1)\cdot n}{2}\big)\right).$$
\end{lem}

\begin{proof}
Let $X=(X_1,\ldots,X_n)\in\{0,1\}^n$ be a random vector, where $X_1,\ldots,X_n$ are $n$ independent Bernoulli random variables such that for $1\le i\le n$, $\Pr[X_i=1]=p$ and $\Pr[X_i=0]=1-p$. For the ease of notation, let $\mu=np(1-p)$. By the law of conditional probability,
\begin{align*}
  \Pr[d(X)\le (1-\ee)\mu|wt(\vx)=pn]=\fr{\Pr[d(X)\le (1-\ee)\mu\wedge wt(\vx)=pn]}{\Pr[wt(\vx)=pn]}=\fr{\bi{n}{pn}-|B|}{\bi{n}{pn}}.
\end{align*}
 To prove the lemma, it is enough to show that
\begin{align*}
    \Pr[d(X)\le (1-\ee)\mu|wt(\vx)=pn]\le n^{3/2}\exp\big(-\fr{\Omega_{p,\ee}(1)\cdot n}{2}\big).
\end{align*}
Recall that $d(X)=\min\{1\le i\le n-1:d(X,\pi_i(X))\}$, where $\pi_i(X)=(X_{i+1},\ldots,X_i)$. Then,
\begin{align*}
&\Pr[d(X) \le (1-\ee)\mu|wt(X)=pn]\\
    = &\Pr[\exists~1\le i\le n-1~\text{s.t.}~d(X,\pi_i(X))\le (1-\ee)\mu|wt(X)=pn]\\
    \le&\sum_{i=1}^{n-1} \Pr[d(X,\pi_i(X))\le (1-\ee)\mu|wt(X)=pn]\\
    =&\sum_{i=1}^{n-1} \fr{\Pr[d(X,\pi_i(X))\le (1-\ee)\mu\wedge wt(\vx)=pn]}{\Pr[wt(\vx)=pn]}\\
    \le&\sum_{i=1}^{n-1} \fr{\Pr[d(X,\pi_i(X))\le (1-\ee)\mu]}{\Pr[wt(X)=pn]}\\
    \le&\left(\sum_{i=1}^{n-1} \Pr[d(X,\pi_i(X))\le (1-\ee)\mu]\right)\cdot \sqrt{2\pi np(1-p)}\cdot \ell(n),
\end{align*}
with $\ell(n)=\exp(-\fr{1}{12n+1}+\fr{1}{12pn}+\fr{1}{12(1-p)n})$, where the first inequality follows by the union bound, and the last inequality follows from the fact that
 $$\Pr[wt(X)=pn]=\binom{n}{pn}p^{pn}(1-p)^{(1-p)n}$$
 and Stirling's approximation
 $$\binom{n}{pn}p^{pn}(1-p)^{(1-p)n}\ge \fr{1}{\sqrt{2\pi np(1-p)}}\cdot\exp(\fr{1}{12n+1}-\fr{1}{12pn}-\fr{1}{12(1-p)n}).$$

According to the discussion above, to prove the lemma it suffices to show that for every fixed $1\le i\le n-1$,
$$\Pr[d(X,\pi_i(X))\le (1-\ee)\mu]\le\exp(-\fr{\Omega_{p,\ee}(1)\cdot n}{2}).$$
Similarly to the proof of \cref{lem: dle}, we will apply \cref{thm: mc} with $f(x_1,\ldots,x_n):=d(\vx,\pi_i(\vx))=n-\sum_{j=1}^{n}\delta(x_j,x_{j+i})$. As for every $i\neq j$, $\E[\delta(X_i,X_j)]=p^2+(1-p)^2$, by the linearity of expectation we have that $\E[f(X_1,X_2,\ldots,X_n)]=2np(1-p)=2\mu$.
Given \eqref{eqa:lipchitz}, applying \cref{thm: mc} to $f$ yields that
\begin{align*}
    \Pr\left[d(X,\pi_i(X))-2\mu \le -(1+\ee)\mu\right]\le \exp \big(- \frac{(1+\ee)^2p^2(1-p)^2n}{2} \big)=\exp\big(-\fr{\Omega_{p,\ee}(1)\cdot n}{2}\big),
\end{align*}
and then we have arrived at the desired conclusion.
\end{proof}

Similarly to the discussion in \cref{rem:large-V}, the set $B$ obtained by \cref{lem: dle-cw-b} can be partitioned into a family $V$ of pairwise disjoint equivalent classes, where
\begin{align}\label{eqa:|V|-2}
    |V|=\frac{|B|}{n}\ge \frac{1}{n}\cdot \binom{n}{pn}\left(1-n^{3/2}\exp(-\fr{\Omega_{p,\ee}(1)n}{2})\right).
\end{align}

To prove \cref{thm:bi-cw-hc-code}, we will construct an auxiliary graph on the vertex set $V$, as detailed below.

\begin{lem}\label{lem:implyThm2}
  Given $d\in[\tau n, (1-\ee)np(1-p)]$, let $G_{OOC}=(V,E)$ be a graph constructed as below, where $V$ is the family of equivalent classes that partition $B$, and two distinct vertices $C(\vx),C(\vy)\in V$ are connected if and only if $d(C(\vx),C(\vy))\le d-1$. Then the following holds:
  \begin{itemize}
      \item [{\rm (i)}] for every independent set $I\subseteq V$, the set of vectors $ C:=\bigcup_{C(\vx)\in I} C(\vx)$ forms an $(n,n|I|,d;pn)$-OOC;
      \item [{\rm (ii)}] the maximum degree of $G_{OOC}$ is at most $D:=Vol(n,d-1;pn)$;
      \item [{\rm (iii)}] for every vertex of $G_{OOC}$, its neighborhood induces at most $\frac{D^2}{K}$ edges, where $K=e^{\Omega_{\tau,\ee}(1)\cdot n}$.
  \end{itemize}
\end{lem}

As the proof of \cref{lem:implyThm2} is very similar to that of \cref{lem:implyThm1}, we will put it in Appendix.

Next we present the proof of \cref{thm:bi-cw-hc-code}.

\begin{proof}[Proof of \cref{thm:bi-cw-hc-code}]
Let $G_{OOC}=(V,E)$ be the graph defined in \cref{lem:implyThm2}. Then, applying \cref{lem: graph} in concert with \eqref{eqa:|V|-2} and \cref{lem:implyThm2} (ii), (iii) yields that
\begin{align*}
  \alpha(G_{OOC})&\ge(1-o(1))\cdot\fr{|V|}{D}\cdot\ln(\min\{D,K\})\\
  &\ge (1-o(1))\cdot \frac{1}{n}\cdot\bi{n}{pn}\cdot\fr{1}{Vol(n,d-1;pn)}\cdot\ln(\min\{Vol(n,d-1;pn),e^{\Omega_{\tau,\ee}(n)}\})\\
  &\ge \fr{c\bi{n}{pn}}{Vol(n,d-1;pn)},
\end{align*}
where $c$ is a constant independent of $n$. Therefore, according to \cref{lem:implyThm2} (i), there exist $(n,M,d;pn)$-OOCs with $M\ge n\cdot \alpha(G_{OOC})$, completing the proof of the theorem.
%
\end{proof}

\section{Conclusions}\label{sec:con}

In this paper we present improved GV-type bounds for hopping cyclic codes and optical orthogonal codes, which in turn give improved GV-type bounds for non-linear cyclic codes and non-linear constant weight cyclic codes, respectively. Several interesting problems remain open.


\begin{qu}
It is known that for prime power $q\ge 49$, there exist a class of codes, called algebraic geometry codes, that can significantly improve the GV bound \eqref{eqa:GV} by an exponential power $\exp(\Omega(n))$ ( see \cite{Hoholdt-AG-book,Tsfasman-AG-paper}). Does such an improvement hold also for (non-linear) cyclic codes? What about constant weight cyclic codes?
\end{qu}

\begin{qu}
All of the results mentioned in this paper are non-constructive. Can we explicitly construct a class of asymptotically good (non-linear) cyclic codes? Or more precisely, can we present Zyablov-type bounds for (non-linear) cyclic codes? What about constant weight cyclic codes?
\end{qu}


\section*{Acknowledgements}

The first two authors would like to thank Prof. Yanxun Chang for an inspiring talk on constant weight cyclic codes.

\normalem

\section{Appendix: proof of \cref{lem:implyThm2}}\label{sec:app}

\begin{proof}[Proof of \cref{lem:implyThm2}]
Note that given the proofs of  the first two items of \cref{lem:implyThm1}, \cref{lem:implyThm2} (i) and (ii) can be proved easily, so we omit both proofs for clarity. To prove (iii), let $\Gamma,S,$ and $T$ be defined as in the proof of \cref{lem:implyThm1}. Then, it suffices to show that \cref{claim:sparse-graph-1} holds also under the setting of \cref{lem:implyThm2}.

\begin{claim}[restatement of \cref{claim:sparse-graph-1}]
There exists some $K=e^{\Omega_{\tau,\ee,p}(1)\cdot n}$ such that $|S|\le D/K$, and for every vertex $C(\vy)\in T$, $\deg_{\Gamma}(C(\vy))\le D/K$.
\end{claim}

To prove the first part of the claim, note that
\begin{align}\label{eqa:S2}
    \frac{|S|}{D}=\fr{|\{C(\vy)\in V(\Gamma):d(C(\vx),C(\vy))\le d-\tau n/2\}|}{D}\le\fr{Vol(n,d-\tau n/2;pn)}{Vol(n,d-1;pn)}\le2e^{-\Omega_{\tau,\ee}(1)n},
\end{align}
where the first inequality follows from the same technique that proves the first inequality in \eqref{eqa:S}, and the second inequality follows by the first equation of \cref{lem:vol cw}.

As for the second part of the claim, using the same strategy that proves \eqref{eqa:T} one can show that for every $C(\vy)\in T$,

\begin{equation}\label{eqa:T2}
  \begin{aligned}
  \fr{\deg_{\Gamma}(C(\vy))}{D}&=\fr{|N_{\Gamma}(C(\vy))|}{Vol(n,d-1;pn)}\le\fr{\sum_{0\le i,j\le n-1} |B(\pi_i(\vx),d-1)\cap B(\pi_j(\vy),d-1)|}{Vol(n,d-1;pn)}\\
  &=\sum_{0\le i,j\le n-1}\frac{B(\pi_i(\vx),d-1;pn)\cap B(\pi_j(\vy),d-1)}{Vol_q(n,d-1;pn)}\\
  &\le\sum_{0\le i,j\le n-1} e^{-\Omega_{\ee}(1)\cdot (d(\pi_i(\vx),\pi_j(\vy))+d(\pi_i(\vx),\pi_j(\vy))^2/(d-1))}\\
  &\le n^2\cdot  e^{-\Omega_{\tau,\ee,p}(1)\cdot n}=e^{-\Omega_{\tau,\ee,p}(1)\cdot n},
\end{aligned}
\end{equation}
where the last two inequalities follows from the second equation of \cref{lem:vol cw} and the assumption that $d(\pi_i(\vx),\pi_j(\vy))\ge d(C(\vx),C(\vy))\ge d\ge \tau n$.

The claim follows straightforwardly by combining \eqref{eqa:S2} and \eqref{eqa:T2}.
\end{proof}

\end{document}